\documentclass[12pt,doublespacing]{article}

\usepackage{amsfonts}
\usepackage{amsmath}
\usepackage{amssymb}
\usepackage{mathrsfs}

\newtheorem{theorem}{Theorem}
\newtheorem{proposition}[theorem]{Proposition}
\newtheorem{lemma}[theorem]{Lemma}
\newtheorem{definition}[theorem]{Definition}
\newtheorem{corollary}[theorem]{Corollary}
\newtheorem{remark}[theorem]{Remark}

\newenvironment{proof}[1][Proof]{\noindent\textbf{#1.} }{{\hfill}{\small$\square$}}\textheight21cm \topmargin0mm\textwidth156mm
\newcommand\unit{\hbox{\rm 1\kern-2.8truept l}}

\newcommand\Lform{{\mathcal{L}}\kern-7.56pt\raise1.0pt\hbox{$-$}}
\newcommand{\initsp}{\mathsf{h}}

\begin{document}

\title{A CYCLE DECOMPOSITION AND ENTROPY PRODUCTION FOR CIRCULANT QUANTUM MARKOV SEMIGROUPS}

\author{JORGE R. BOLA$\tilde{\textrm{N}}$OS-SERVIN $^{*}$, ROBERTO QUEZADA$^{**}$ }

\maketitle

\noindent{Universidad Aut\'onoma Metropolitana, Iztapalapa Campus \\
Av. San Rafael Atlixco 186, Col. Vicentina\\ 09340 Iztapalapa D.F., Mexico.\\
$^*$E-mail: kajito@gmail.com \\
 $^{**}$E-mail: roqb@xanum.uam.mx } \\ \\

\noindent{Submitted to: \\ ``Infinite Dimensional Analysis, Quantum Probability and Related To\-pi\-cs"}

\begin{abstract}  We propose a definition of cycle representation for Quantum Markov Semigroups (qms) and Quantum  Entropy Production Rate (QEPR) in terms of the $\rho$-adjoint. We  introduce the class of circulant qms, which admit non-equilibrium steady states but exhibit symmetries that allow us to compute  explicitly the QEPR, gain a deeper insight into the notion of cycle decomposition and prove that quantum detailed balance holds if and only if the QEPR equals zero.
\end{abstract}

\noindent{\textbf{Keywords:} Non-equilibrium steady state, circulant quantum Markov semigroup, quantum cycle representation, entropy production rate, weighted detailed balance.} \\  \\
\noindent{\textbf{{AMS Subject Classification:} 46L55, 82C10, 60J27}

\section{Introduction}
The notion of equilibrium state of physical systems is well understood and there exist several conditions that  characterize such states, detailed balance and zero entropy production among them. For classical Markov chains the equivalence of these two equilibrium criteria has been proved by Qian et al.\cite{qian} using Kalpazidou's cycle representation for Markov chains\cite{kalpaz}. Non-equilibrium state is a much more subtle notion, since there are a huge variety of behaviors involved in it.  

This work is aimed at contributing to the program outlined in Reference\cite{a-f-q}, namely, \textit{to look for some interesting Gorini-Kossakowski-Sudarshan and Lindblad (GKSL) generators with properties that are rich enough to go beyond the equilibrium situation, but concrete enough to allow explicit study and, in some cases, explicit solutions}. We define Quantum Entropy Production Rate (QEPR) for qms in terms of the $\rho$-adjoint and discuss its connection with Fagnola-Rebolledo's\cite{fag-reb} definition. We propose a definition of cycle representation for Completely Positive (CP) maps and GKSL generators, discussing its connections with the QEPR. To test and illustrate the above notions, we introduce the class of circulant qms that admit non-equilibrium steady states but exhibit pretty symmetries which allow us explicit computation of the QEPR. The symmetry properties of our semigroups arise from an abelian group structure on the state space of the associated classical Markov chain. 

Section 2 is a brief review of quantum detailed balance and its extensions. Our QEPR definition along with some basic properties are discussed in Section \ref{entropy-pro-qms}. A brief review of cycles and passage matrices is made in Section \ref{cycles} . In Section \ref{circulant-matrices} we show how the transition probability matrix of a Markov chain on a finite abelian group is a circulant matrix, which is the leading concept of this article. Section \ref{circ-qms} offers a quantum generalization of the former, named circulant operator, and we use it to define circulant qms; also in this section we propose a definition of quantum cycle representation for CP maps and GKSL generators. In Section \ref{QEPR} and {\ref{QEPR-EPR-compare}, both QEPR and classical EPR are  explicitly computed and compared for a circulant qms and its classical restriction using a diagonal invariant state. The remaining invariant states and its QEPR are studied in Section \ref{other-invariant-states} .

\section{Preliminaries} 

\subsection{Quantum detailed balance}

For uniformly continuous qms on ${\mathcal B}(\initsp)$, with $\initsp$ a separable Hilbert space, a notion of detailed balance was introduced first by A\-lic\-ki\-\cite{Ali} and Frigerio-Gorini-Kossakowski-Verri\cite{F-G-K-V}. Indeed, a qms with GKSL generator ${\mathcal L}$ satisfies a quantum detailed balance condition in the sense of Ref.\cite{Ali,F-G-K-V} with respect to a stationary state $\rho$ (i.e., $tr\big(\rho{\mathcal L}(x)\big)=0, \; \forall \; x\in{\mathcal B}(\initsp)$), if there exists an operator $\tilde{\mathcal L}$ on $\mathcal B(\initsp)$ and a self-adjoint operator $K$ on $\initsp$ such that for all $x,y\in \mathcal B(\initsp)$ the following relations hold:
\begin{eqnarray}\label{0-detailed-balance}
\begin{aligned}
& tr(\rho\tilde{\mathcal L}(x)y)=tr(\rho x \mathcal L (y)), \\  
& \tilde{\mathcal L}(\cdot) - \mathcal L(\cdot) =2i [K, \cdot].
\end{aligned}
\end{eqnarray} 

The operator $\tilde{\mathcal L}$, is called the $\rho$-adjoint of ${\mathcal L}$. For a wide class of GKSL generators, including those deduced from the stochastic limit of quantum theory, the $\rho$-adjoint coincides with the time-reversed  generator if quantum detailed balance holds. Therefore, $\tilde{\mathcal L}$ can be considered as an extension of the time-reversed GKSL generator to the non-equilibrium situation and we expect that simple non-equilibrium situations should appear when studying the difference between ${\mathcal L}$ and $\tilde{\mathcal L}$, see Accardi-Fagnola-Quezada\cite{a-f-q} and the references therein. 

Other notions of quantum detailed balance have been introduced by Fagnola and Umanit$\grave{\textrm{a
}}$\cite{franco2,franco3}. The main idea is to separate the invariant state $\rho$ into two pieces or, equivalently, define the $\rho$-adjoint using the inner product $\langle a,b\rangle_s =tr(\rho^{1-s} a^*\rho^sb)$ for $0\leq s\leq\frac{1}{2}$, and replace relations (\ref{0-detailed-balance}) by 
\begin{eqnarray*} \label{s-detailed-balance}
\begin{aligned}
& tr(\rho^{1-s}\tilde{\mathcal L}(x)\rho^{s}y)=tr(\rho^{1-s} x \rho^{s} \mathcal L (y)), \\ &
\tilde{\mathcal L}(\cdot) - \mathcal L(\cdot) =2i [K, \cdot]. 
\end{aligned}
\end{eqnarray*} Due to the non-commutativity, these two definitions are not equivalent in general. Clearly, detailed balance in the sense of (\ref{0-detailed-balance}) corresponds with the case $s=0$ in (\ref{s-detailed-balance}). 

Notice that if $\rho$ is a stationary state for ${\mathcal L}$, hence with $y=\unit$ in (\ref{0-detailed-balance}) and using that ${\mathcal L}(\unit)=0$ we get \[0=tr\big(\rho x{\mathcal L}(\unit)\big)=tr\big(\rho\tilde{\mathcal L} (x)\big), \; \forall \; x\in{\mathcal B}(\initsp).\] Therefore, $\rho$ is a stationary state also for $\tilde{\mathcal L}$.

\subsection{The $\rho$-adjoint and special representations}
The $\rho$-adjoint (with $s=0$) $\tilde{\mathcal L}$ of a GKSL generator ${\mathcal L}$ is a GKSL generator if and only if  the last one commutes with the modular automorphism of $\rho$, i.e., ${\mathcal L}\circ \sigma_{-i}=\sigma_{-i}\circ {\mathcal L}$, where $\sigma_{-i}(a)=\rho a\rho^{-1}$, see Theorem 8 in Reference\cite{franco3}.  

The Markov generators can be written in the standard Gorini-Kossakowski-Sudarshan 
and Lindblad (GKSL) representation  
\begin{eqnarray}\label{gksl-rep}
{\mathcal L}(x)= i[H, x] - \frac{1}{2} \sum_{k\geq 1} 
\left( L_{k}^{*} L_{k} x - 2 L_{k}^{*} x L_{k} 
+ x L_{k}^{*} L_k \right),
\end{eqnarray} 
where $H, \; L_k \in {\mathcal B}({\mathcal H})$ with 
$H=H^{*}$ and the series $\sum_{k\geq 1} L_{k}^{*} L_k $ 
is strongly convergent. 

Given a normal state $\rho$ on ${\mathcal B}({\mathcal H})$, a GKSL representation (\ref{gksl-rep}) of ${\mathcal L}$  by a bounded self-adjoint operator H and a finite or infinite sequence $(L_k )_{k\geq 1}$ of elements of ${\mathcal B}({\mathcal H})$ such that:
\begin{itemize}
\item[(i)] ${\rm tr}\left(\rho L_{k} \right)=0$ for each $k\geq 1$,
\item[(ii)] $\sum_{k\geq 1} L_{k}^{*}L_{k}$  is a strongly convergent sum,
\item[(iii)] if $\sum_{k\geq 0} |c_{k} |^{2}<\infty$ and $c_0  + \sum_{k\geq 1} {c_k} L_{k} =0$ for complex scalars $(c_k)_{k\geq 0}$, then $c_k = 0$ for every $k \geq 0$,  
\end{itemize} is called \textit{special}. See Theorem 30.16 in Parthasarathy's book\cite{partha} for a proof of the existence of these class of representations.  Special representations are unique up to unitary transformations.

\subsection{Weighted detailed balance}

The notion of weighted detailed balance introduced in Reference\cite{a-f-q}, was aimed at characterizing a class of GKSL generators with properties rich enough to go beyond the equilibrium situation but concrete enough to allow explicit study. In terms of special representations, weighted detailed balance is  stated as follows. 

A uniformly continuous quantum Markov semigroup $({\mathcal T}_{t})_{t\geq 0}$ satisfies a \textit{weighted detailed balance} condition with respect to a faithful invariant state $\rho$, if its generator ${\mathcal L}$ has a special GKSL representation by means of operators $H, L_{k}$, such that here exists a sequence of positive weights $q:=(q_{k})_{k}$ and operators $ {H'}, {L}_{k}^{'}$ of a (possibly another) special representation of ${\mathcal L}$ such that the difference $\tilde{\mathcal L}_{\rho}-{\mathcal L}$ has the structure  
\begin{eqnarray}\label{def-weighted-db}
\tilde{\mathcal L}_{\rho}-{\mathcal L} = -2i [K, \cdot] + \Pi, 
\end{eqnarray} where $K=K^{*}$ is bounded and 
\begin{eqnarray}\label{Pi-weighted-db}
\Pi(x)= \sum_{k} (q_{k} -1) {L}_{k}^{'*}x {L}_{k}^{'}. 
\end{eqnarray}  Quantum detailed balance holds if and only if $q_{k}=1$ for all $k$.

\section{Quantum Entropy Production Rate for quantum Markov semigroups}\label{entropy-pro-qms}
In this section we introduce a notion of Quantum Entropy Production based on the concept of $\rho$-adjoint. As well as detailed balance, our definition depends on which $\rho$-adjoint is used. Our definition is slightly different from the one introduced by Fagnola and Rebolledo\cite{}. Both definitions coincide in the class of circulant quantum Markov semigroups introduced in Section \ref{circ-qms} below.

Assume that ${\mathcal L}$ and its $\rho$-adjoint $\tilde{\mathcal L}$, are GKSL generators of strongly continuous qms ${\mathcal T}$ and $\tilde{\mathcal T}$, respectively, with an invariant state $\rho$. Let ${\mathcal T}_{*t}$ and $\tilde{\mathcal T}_{*t}$  denote the corresponding pre-dual semigroups.

 \begin{definition} For every $t\geq 0$, let $\Omega_{t}$ and $\tilde{\Omega}_{t}$ be the states (density matrices) on ${\mathcal B}( \initsp\otimes \initsp)$, with $\initsp$ a separable Hilbert space, given by 
 \[\Omega_{t}=\Big(	\unit\otimes {\mathcal T}_{*t}\Big)  \big(|\Omega_{\rho} \rangle \langle \Omega_{\rho} |\big) \]  and 
 \[\tilde{\Omega}_{t}=\Big(	\unit\otimes \tilde{\mathcal T}_{*t}\Big)  \big( |\Omega_{\rho} \rangle \langle \Omega_{\rho} | \big), \] where $\Omega_{\rho}= \displaystyle\sum_i  \rho_{i}^{\frac{1}{2}}(e_i \otimes e_i) \in  (h \otimes  h)$, with $(e_{i})_{1\leq i\leq p-1}$ the orthonormal basis of $\rho$ in $\initsp$ . The Quantum Entropy Production Rate of the uniformly continuous qms ${\mathcal T}_{*}$, with respect to the invariant state $\rho$, is given by  \[e_p({\mathcal T_*}, \rho)=\left. \frac{d}{dt}S(\Omega_{t}, \tilde{\Omega}_{t})\right| _{t=0}, \label{quan-ent}\] where the relative entropy of the states $\eta$ and $\rho$ is defined as 
\[S(\eta,\rho)=tr\Big(\eta \log \eta- \eta \log \rho\Big)\] if the nullspace of $\eta$ contains the nullspace of $\rho$ and $\infty$ otherwise.
\end{definition}

As a consequence of Klein's Inequality, see the work of B. Ruskai\cite{ruskai}, the relative entropy of every pair of states $\rho,\eta$ is non-negative
\[S(\eta,\rho)\geq 0. \] Moreover, equality holds if and only if $\eta=\rho$.

In the remaining sections we compute explicitly the Quantum Entropy Production Rate for circulant qms.  
 
 \begin{remark}
 \begin{itemize}
 \item[(i)] In the finite dimensional case $\Omega_{t}$ (resp. $\tilde{\Omega}_{t}$) is the so called Jamio\l kowski\cite{Jamiolkowski}, or Choi-Jamio\l kowski, transform  of the CP map ${\mathcal T}_{*t}\circ T_{\rho}$ (resp. $\tilde{\mathcal T}_{*t}\circ T_{\rho}$), with $T_{\rho}(x)=\rho^{\frac{1}{2}} x \rho^{\frac{1}{2}} $.  
 \item[(ii)]  A simple computation shows that $tr\Big(|\Omega_{\rho}\rangle\langle \Omega_{\rho}|\Big)=tr(\rho)=1$, hence $|\Omega_{\rho}\rangle\langle\Omega_{\rho}|$ is a state on ${\mathcal B}(\initsp\otimes\initsp)$ and $\Omega_{t}$ is well defined.
 \item[(iii)] In comparison to Fagnola-Rebolledo's definition of entropy production rate, we remark that in our definition, the Jamio\l kowski transform is not modified by an anti-unitary operator. Moreover, instead of forward and backward two-point states we use  as forward dynamics the time-dependent state generated by Jamio\l kowski transform of the semigroup $({\mathcal T}_{*t})_{t\geq 0}$ and as a backward dynamics the one associated with its $\rho$-adjoint $(\tilde{\mathcal T}_{*t})_{t\geq 0}$.
 
 \end{itemize}
 
 \end{remark}

\section{Cycles and passage functions} \label{cycles} Let $S$ be a numerable set and $c$ a periodic function from ${\mathbb Z}$ into the set $S$. Following the notations of Qian et al.\cite{qian}, we call the values $c(n)$ of $c$ vertices (or nodes) of $c$,  while the pairs $(c(n), c(n+1))$ are called edges (directed edges or directed arcs)  of $c$. The period of $c$ is the smallest integer $p$ such that $c(n+p)=c(n)$ for all $n\in{\mathbb Z}$. Two periodic functions $c$ and $c'$ are equivalent if one is a translation of the other, i.e., there exists $i\in{\mathbb Z}$ such that $c'(n)=c(n+i)$. The above is an equivalence relation and clearly two equivalent periodic functions have the same vertices and period. A \textbf{directed circuit} is an equivalence class of the above defined equivalence relation. Any directed circuit is determined either by its period $p$ and  any $(p+1)$-tuple $(i_0 , i_1, \cdots, i_{p})$ with $i_{p}=i_0$; or by its period $p$ and $p$ ordered pairs $(i_0, i_1), (i_1, i_2), \cdots , (i_{p-1}, i_{p})$ with $i_{p}=i_0$, where $i_l =c(n+l-1), \; 0\leq i_l \leq p-1$ for some $n\in{\mathbb Z}$.

\begin{definition}
The cycle (or directed cycle) associated with a given directed circuit $c=(i_0, i_1, \cdots, i_{p-1}, i_1), \;  p\geq 1$, with distinct vertices $i_0, i_1, \cdots i_{p-1}$, is the ordered sequence $\hat{c}=(i_0, i_1, \cdots i_{p-1})$.
\end{definition} Every cycle is invariant under cyclic permutation of its vertices. We also use the notation $\hat{c}=(c(0), c(1), \cdots, c(p-1))$ for the cycle associated with the directed circuit $c=(c(0), c(1), \cdots, c(p-1), c(0))$ of period $p$, and use the symbol $c$ for both the directed circuit and the cycle when no confusion is possible.  
For every directed circuit $c=(i_0,i_1,\ldots,i_{p-1},i_0)$, the \textbf{reverse} circuit $c_{-}$ is defined as $c_{-}=(i_0,i_{p-1},\ldots,i_1,i_0)$.

When all points of $c$ are distinct  except for the extremes, then 
\begin{eqnarray*}
\begin{aligned}
J_c(i,j)=\left\{ \begin{array}{ll}
         1 & \mbox{if $(i,j)$ is an edge of $c$};\\
        0 & \mbox{otherwise}.\end{array} \right.
\end{aligned}
\end{eqnarray*}       
        
Hence, every cycle $c$ has associated an unique matrix $J_{c}=(J_{c}(i,j))$, in some complex matrix space, called \textbf{passage matrix} of $c$. \\

\textbf{Example.} If $c_0=(0123)$ and $c_1=(0312)$, then 

\begin{align*}\begin{array}{cr}J_{c_0}=\left(\begin{array}{ccccc}0&1&0&0\\ 0&0&1&0\\ 0&0&0&1 \\1&0&0&0\end{array}\right) \textrm{and} & J_{c_1}=\left(\begin{array}{ccccc} 0&0&0&1\\0&0&1&0\\1&0&0&0\\0&1&0&0\end{array}\right) \end{array}\end{align*}

The passage matrix $J_{c_0}$ of the full length cycle $c_0=(0,1,2,\ldots,p-1)$ is often called \textit{the primary permutation matrix} and the cycle $c_0$ \textit{the primary cycle}. From now on $J_{p}$ will denote the primary permutation matrix in a $p\times p$ complex matrix space. Notice that given any cycle $c=(c(0),c(1), \cdots, c(p-1))$, its passage matrix can be written in terms of the canonical basis $\{|e_{i}\rangle\langle e_{j}| : 0\leq i,j\leq p-1\}$ of the $p\times p$ complex matrix space as \[J_{c}=\sum_{i=0}^{p-1} |e_{c(i)}\rangle\langle e_{c(i+1)}|,\] where $\{e_{j}\}_{0\leq j\leq p-1}$ is the canonical basis of ${\mathbb C}^{p}$.
Notice that $J_c$ moves the canonical basis of ${\mathbb C}^p$ according to the cycle $c$, i.e., $J_c e_{c(i)} = e_{c(i-1)}$ for all $i$. So, the primary permutation matrix $J_{p}$ is, in fact, the left shift operator for the canonical basis in ${\mathbb C}^p$.

\section{Circulant matrices}\label{circulant-matrices}
\subsection{Markov chains on finite groups} Let $(G,\circ)$ be a finite group. Unless otherwise specified, we let $p=|G|$ and denote by $hg$ the product $h\circ g$, $h,g\in G$. Given a probability distribution $\mu$ on $G$, the transition probabilities \[p(g, hg)=\mu(\{h\}),\] define a discrete time Markov chain on $G$. \\ 

\textbf{Example 1.} Consider the cyclic group ${\mathbb Z}_p =\{0,1, \cdots, p-1\}$ and any distribution probability $\alpha=\{\alpha_0, \alpha_1, \cdots, \alpha_{p-1}\}$ on ${\mathbb Z}_p$. Then the transition probability matrix is the circulant matrix  
 \begin{eqnarray*}
 \begin{aligned}
 A=\textrm{circ}(\alpha_0,\alpha_{1}, \alpha_{2}, \cdots , \alpha_{p-1} ) =\left(\begin{array}{ccccc}\alpha_0&\alpha_{1}&\alpha_{2}& \cdots & \alpha_{p-1} \\
                                      \alpha_{p-1} & \alpha_0 & \alpha_{1} & \cdots &\alpha_{p-2} \\
                                      \alpha_{p-2} & \alpha_{p-1} & \alpha_0 & \cdots & \alpha_{p-3} \\
                                      \vdots & \vdots & \vdots & \ddots & \vdots \\
                                      \alpha_{1} & \alpha_{2} & \alpha_{3} & \cdots & \alpha_0\end{array}\right).
                                      \end{aligned}
                                      \end{eqnarray*} 
                                      
                                       Notice that $A$ is a convex linear combination of powers of the primary permutation matrix $J_{p}=\sum_{j=0}^{p-1}|e_{j}\rangle\langle e_{j+1}|$;  indeed, 
                                       \begin{eqnarray}\label{birkhoff-rep}
                                       \begin{aligned}
                                       A=\sum_{j=0}^{p-1}\alpha_{j}J_{p}^{j}.
                                       \end{aligned}
                                       \end{eqnarray}

\textbf{Example 2.}  Let $G$ be the abelian group $G={\mathbb Z}_{p}\times{\mathbb Z}_{q}$ where the symbol $\times$ denotes direct product, with $p, \; q\geq 2$. We set the lexicographic order in ${\mathbb Z}_{p}\times{\mathbb Z}_{q}$ and take $\alpha=\{\alpha(0,0), \cdots, \alpha(0,q-1), \alpha(1, 0), \cdots, \alpha(1,q-1), \cdots , \alpha(p-1, 0), \cdots,\alpha(p-1, q-1)\}$ any probability distribution on $G$. One can easily see that the corresponding transition probability matrix is the block circulant matrix \[R=\textrm{circ}(R_{0}, R_{1}, \cdots, R_{p-1}),\] with circulant blocks 
\[R_{i}=\left(\begin{array}{ccccc}\alpha(i,0)&\alpha(i,1)&\alpha(i,2)& \cdots & \alpha(i, q-1) \\
                                      \alpha(i,q-1) & \alpha(i,0) & \alpha(i,1) & \cdots &\alpha(i,q-2) \\
                                      \alpha(i,q-2) & \alpha(i, q-1) & \alpha(i,0) & \cdots & \alpha(i,q-3) \\
                                      \vdots & \vdots & \vdots & \ddots & \vdots \\
                                      \alpha(i,1) & \alpha(i,2) & \alpha(i,3) & \cdots & \alpha(i,0)
                                      \end{array}\right), i=0,1,\cdots, p-1. \] 

The above matrix $R$ is a convex linear combination of tensor products of powers of the primary permutation matrices $J_{p}=\sum_{i=0}^{p-1}|e_{i}\rangle\langle e_{i+1}|$ and $J_{q}=\sum_{j=0}^{q-1}|e_{j}\rangle\langle e_{j+1}|$, indeed, 
\begin{eqnarray}\label{birkhoff-rep-product}
\begin{aligned}
R=\sum_{0\leq i\leq p-1,\, 0\leq j\leq q-1}\alpha(i,j) J_{p}^{i}\otimes J_{q}^{j}.
\end{aligned}
\end{eqnarray} 

\begin{remark}
Due to Birkhoff's Theorem, every bi-stochastic matrix is a convex  linear combination of permutation matrices. Notice that (\ref{birkhoff-rep}) and (\ref{birkhoff-rep-product}) are Birkhoff's representations of bi-stochastic circulant and block-circulant matrices, respectively.
\end{remark}

%---------------------------------------------------------

\subsection{Diagonalization of circulant matrices}\label{diag-circ-matrix}

The discrete (or quantum) Fourier transform on ${\mathbb C}^{p}$ is the unitary operator defined by means of 
\[F_{p}=\frac{1}{\sqrt{p}}\sum_{0\leq j,k\leq p-1}\omega_{p}^{kj}|e_{k}\rangle \langle e_{j}|,\] where $\omega_{p}$ is a primitive $p$-th root of unity and $\{e_{j}\}_{0\leq j\leq p-1}$ is the canonical basis of ${\mathbb C}^p$. Before proving the discrete Fourier transform diagonalizes $J_{p}$ we will need the next ortoghonality relation between the $p$-th roots of unity.

\begin{proposition} For every pair $i,k\in\{0,1,\ldots,p-1\}$ 
\begin{align*}\sum_{l=0}^{p-1}\omega_{p}^{(k-i)l}=p\delta_{ik}\end{align*}
\end{proposition}

\begin{proof}
Fix $i,k\in\{0,1,\ldots,p-1\}$. Any primitive $p$-th root of unity $\omega_{p}$ satisfies

\begin{align*}\omega_{p}^{k}\overline{\omega}_{p}^i\frac{1}{p}\sum_l \omega_{p}^{lk}\overline{\omega}_{p}^{il} = \frac{1}{p}\sum_l \omega_{p}^{(l+1)k}\overline{\omega}_{p}^{(l+1)i}=\frac{1}{p}\sum_l \omega_{p}^{lk}\overline{\omega}_{p}^{il},\end{align*} therefore $(\omega_{p}^{(k-i)}-1)\dfrac{1}{p}\displaystyle\sum_l\omega_{p}^{(k-i)l}=0 $. Since $\omega_{p}^{(k-i)}-1 =1-\delta_{ik}$ the conclusion follows.
\end{proof}

\begin{lemma}\label{diagonal-circ} Let $Z_p =\textrm{diag}(1, \overline{\omega}_{p}, \overline{\omega}_{p}^{2}, \cdots, \overline{\omega}_{p}^{p-1})$, then 
\begin{itemize}
\item[(i)] $F_{p}J_{p}F_{p}^{*}=Z_p,$
\item[(ii)] $(F_{p}\otimes F_{q})(J_{p}\otimes J_{q})(F_{p}\otimes F_{q})^{*}=Z_{p}\otimes Z_{q}.$
\end{itemize}
\end{lemma}

\begin{proof} Direct computations show that 
\[F_{p}J_{p}=\frac{1}{\sqrt{p}}\sum_{k,l}\omega_{p}^{kl}|e_{k}\rangle\langle e_{l+1}|,\]  
\begin{eqnarray*}
\begin{aligned}
F_{p}J_{p}F_{p}^{*} &= \frac{1}{p} \Big(\sum_{k,l}\omega_{p}^{kl}|e_{k}\rangle\langle e_{l+1}|\Big)\Big(\sum_{i,j}\overline{\omega}_{p}^{ij}|e_{j}\rangle\langle e_{i}|\Big)= \frac{1}{p}\sum_{i,k,l}\omega_{p}^{kl-(l+1)i} |e_{k}\rangle\langle e_{i}| \\  &= \frac{1}{p}\sum_{i,k}\Big(\sum_{l}\omega_{p}^{(k-i)l}\Big)\omega_{p}^{-i}|e_{k}\rangle\langle e_{i}| = \sum_{k}\overline{\omega}_{p}^{k}|e_{k}\rangle\langle e_{k}| =Z_p.
\end{aligned}
\end{eqnarray*}This proves $(i)$. Item $(ii)$ follows directly from $(i)$.
\end{proof}

Since each circulant matrix can be expressed in terms of the primary permutation matrix $J_p$, it follows that the discrete Fourier transform diagonalizes every circulant matrix as well as block circulant matrices with circulant blocks.

\begin{theorem}\label{diag-block-circ} If $A=\sum_{i}\alpha(i)J_{p}^{i}$ and $B=\sum_{i,j} \alpha(i,j) J_{p}^{i}\otimes J_{q}^{j},$ then 
\begin{itemize}
\item[(i)] \[F_{p}AF_{p}^{*}=\sum_{k}\lambda_{k} |e_{k}\rangle\langle e_{k}|,\] with $\lambda_{k}=\sum_{i}\alpha(i)\overline{\omega}_{p}^{k i}$, and  
\item[(ii)]  \[(F_{p}\otimes F_{q})B(F_{p}\otimes F_{q})^{*}=\sum_{k,l} \lambda_{k,l} |e_{k}\otimes e_{l}\rangle\langle e_{k}\otimes e_{l}|,\] with $\lambda_{k l}=\sum_{i,j}\alpha(i,j)\overline{\omega}_{p}^{ik}\overline{\omega}_{q}^{jl}$.
\end{itemize}
\end{theorem}

\begin{proof}
Using the above Lemma \ref{diagonal-circ}, a direct computation shows that  
\[F_{p}AF_{p}^{*} = \sum_{i} \alpha(i)F_{p}J_{p}^{i}F_{p}^{*}= \sum_{i}\alpha(i)\sum_{k}\overline{\omega}_{p}^{ik}|e_{k}\rangle\langle e_{k}|=\sum_{k}\Big(\sum_{i}\alpha(i)\overline{\omega}_{p}^{ki}\Big) |e_{k}\rangle \langle e_{k}|.\] This proves $(i)$.

Now observe that 

\begin{eqnarray*}
\begin{aligned}
(F_{p}\otimes F_{q})B(F_{p}\otimes F_{q})^{*} & =\sum_{i,j}\alpha(i,j)(F_{p}J_{p}^{i} F_{p}^{*})\otimes(F_{q}J_{q}^{j} F_{q}^{*}) \\ & = \sum_{i,j}\alpha(i,j)\Big(\sum_{k}\overline{\omega}_{p}^{ik} |e_{k}\rangle\langle e_{k}\Big)\otimes\Big(\sum_{l}\overline{\omega}_{q}^{jl} |e_{l}\rangle \langle e_{l}|\Big) \\ & = \sum_{k,l}\Big(\sum_{i,j}\alpha(i,j)\overline{\omega}_{p}^{ik}\overline{\omega}_{q}^{jl}\Big)|e_{k}\otimes e_{l}\rangle\langle e_{k}\otimes e_{l}|.
\end{aligned}
\end{eqnarray*} This finishes the proof. 
\end{proof}

\begin{corollary}\label{exp-block-circ} With the notations in the above theorem we have 
\begin{itemize}
\item[(i)] \[e^{tA} = \frac{1}{p}\sum_{j,l}\Phi_{l-j}(t)|e_{j}\rangle\langle e_{l}|,\] with $\Phi_{m}(t)=\sum_{k}\omega_{p}^{m k} e^{t\lambda_{k}},$ and 
\item[(ii)] \[e^{t B}= \frac{1}{pq}\sum_{i,j,m,n}\Phi_{m-i, n-j}(t)|e_{i}\otimes e_{j}\rangle \langle e_{m}\otimes e_{n}|,\] with $\Phi_{i,j}(t)=\sum_{k,l}\omega_{p}^{ik}\omega_{q}^{jl} e^{t\lambda_{kl}}.$ 
\end{itemize}
\end{corollary}
\begin{proof} The result of the above theorem and a direct computation show that 
\begin{eqnarray*}
\begin{aligned}
e^{tA}= F_{p}^{*}\textrm{diag}(e^{t\lambda_{k}})F_{p}= \frac{1}{p}\sum_{j,l}\Big(\sum_{k}\omega_{p}^{(l-j)k}e^{t\lambda_{k}}\Big)|e_{j}\rangle \langle e_{l}| =\frac{1}{p}\sum_{j,l}\Phi_{l-j}(t)|e_{j}\rangle \langle e_{l}|.
\end{aligned}
\end{eqnarray*} This proves $(i)$.

In a similar way we see that 
\begin{eqnarray*}
\begin{aligned}
e^{tB} & = (F_{p}\otimes F_{q})^{*}\textrm{diag}(e^{t\lambda_{kl}})(F_{p}\otimes F_{q})\\ &
= \frac{1}{pq}\sum_{i,j,k,l,m,n} \overline{\omega}_{p}^{ik}\overline{\omega}_{q}^{jl}e^{t\lambda_{kl}}\omega_{p}^{km}\omega_{q}^{ln}|e_{i}\otimes e_{j}\rangle \langle e_{m}\otimes e_{n}| \\ & = \frac{1}{pq} \sum_{i,j,m,n}\Phi_{m-i, n-j}(t) |e_{i}\otimes e_{j}\rangle \langle e_{m}\otimes e_{n}|.
\end{aligned}
\end{eqnarray*}
\end{proof}

%--------------------------------------------------------------------------------

\section{Circulant quantum Markov semigroups}\label{circ-qms}

\subsection{Circulant completely positive maps}\label{circ-channels}
We first recall that a $p\times p$ complex matrix $A$ is called \textbf{reducible} if there exists a permutation matrix $P$ such that  
\[PAP^{*} = \left(\begin{array}{cc}B&C \\
                                      0 & D 
                                      \end{array}\right),\] 
where $B$ and $D$ are square matrices of order at least $1$. A matrix is called \textbf{irreducible} if it is not reducible. It is well known, see for instance Theorem 5.18 in Zhang's book\cite{fuzhen}, that every irreducible $p\times p$ permutation matrix $A$ is permutation similar to the primary permutation matrix $J_{p}$, i.e., there exists a permutation matrix $P$ such that $A=PJ_{p}P^{-1}$. 

\begin{lemma}\label{lema1} For every irreducible permutation matrix $J\in\mathcal M_p (\mathbb C)$ there exists a unique cycle $c$ of maximal length such that $J=J_c$ is the passage matrix of $c$.
\end{lemma}
\begin{proof} Being an irreducible permutation matrix, $J$ is permutation similar to $J_{p}$, i.e., there exists a permutation matrix $P$ such that $J=PJ_{p}P^{-1}$. Therefore for any element of the canonical basis $\{e_i\}$ of $\mathbb C^n$ we have 
\begin{align*} JP e_i= PJ_{p} e_i =Pe_{i-1}.\end{align*} Define a unique cycle $c$ by means of the permutation $P$ taking $e_{c(i)}=P e_{i-1}$. Clearly $J=J_c$ since $J e_{c(i)} =e_{c(i-1)}$. 
\end{proof}

\begin{lemma} \label{lemma-subspaces} Let $B_0,B_1,\ldots B_{p-1}$ be $p$-dimensional subspaces mutually orthogonal with respect to the Hilbert-Schmidt inner product in $\mathcal M_p (\mathbb C)$, with $B_0$ the subspace of all diagonal matrices. If $J_c$ is the passage matrix of any cycle $c$ of maximal length, then 
\begin{eqnarray}\label{subspaces}
\begin{aligned}
& J_c B_l =B_{l+1}, \, \forall \; 0\leq l\leq p-1 \Longleftrightarrow \\  & B_l=span\{|e_{c(k)}\rangle \langle e_{c(k+l)}| k=0,1,\ldots,n-1\} \; \forall \;  0\leq l\leq p-1, 
\end{aligned}
\end{eqnarray} where the sums in the indices $k, l$ is modulus $p$.
 \end{lemma}
 \begin{proof} Clearly condition on the left hand side of (\ref{subspaces}) is sufficient for $J_c B_l=B_{l+1}$ for all $0\leq l\leq p-1$. Let us proof the necessity  by induction on $p\geq 1$. For $p=1$ the condition on the left hand side clearly holds. Now, assuming that the condition holds for any $1\leq p$ and let us proof that it holds for $p+1$. We have $J_{c}|e_{c(i)}\rangle\langle e_{c(i+l)}|=|e_{c(i-1)}\rangle \langle e_{c(i+l)}|\in B_{l+1}, \, \forall \,  0\leq i\leq p,  0\leq l\leq p $ by assumption, where the sums in $c(i-1), c(i+l)$  is modulus $p+1$. Hence we have that $B_{l}=\textrm{span}\{|e_{c(i)}\rangle\langle e_{c(i+1)}| : \, 0\leq i\leq p\},  \, \forall \, 0\leq l\leq p$ since these $p+1$ vectors are linearly independent and $B_{l}$ is ($p+1$)-dimensional. \; 
 \end{proof}

\begin{definition} A linear operator $\Phi:\mathcal M_p (\mathbb C)\to\mathcal M_p (\mathbb C)$ is called circulant map (or circulant quantum channel) if there exist $p$-dimensional subspaces $B_l$, $l=0\ldots,p-1$, mutually orthogonal with respect to the Hilbert-Schmidt inner product with $B_0=span\{|e_k\rangle \langle e_k|: k=0,\ldots,p-1\}$, invariant under the action of $\Phi$, such that 
\begin{itemize}
\item[(i)] $\mathcal M_p (\mathbb C)= \bigoplus_{l=0}^{p-1} B_l$.

\item[(ii)] there exists an irreducible permutation matrix $J\in \mathcal M_p(\mathbb C)$ such that $J B_l =B_{l+1}$ sum modulus $p$.

\item[(iii)] If $c$ is the cycle associate with $J$ by Lemma \ref{lema1}, then under the isomorphism from  $B_l$ into  $\mathbb C^n$ defined by $|e_{c(k)}\rangle \langle e_{c(k+l)}|\mapsto e_{c(k)} $ we have that 
\begin{align*} \Phi(|e_{c(k)}\rangle \langle e_{c(k+l)}|)\mapsto e_{c(k)}Q 
\end{align*}  where $Q$ is a circulant $p\times p$ matrix.
\end{itemize}

\end{definition}

\textbf{Example.} Let $c$ be any cycle of maximal length in $\{0,1,\cdots,p-1\}$, then the CP linear map defined by \[\Phi_c (x)=\sum_{k=0}^{p-1} \gamma(p-k) J^{k}_c x J^{*k}_c,\] for some $\gamma(j)\geq 0$ is a circulant CP map. 
Let us define the subspaces $B_l=\textrm{span}\{|e_{c(k)}\rangle \langle e_{c(k+l)}|: k=0,1,\ldots,p-1\}$ for $l=0,\ldots,p-1$, clearly condition $(i)$ in the above definition holds. With $J=J_c$ in the above definition, let us prove the invariance of the subspaces $B_l$'s. For any $x^{(l)}=\sum_{j=0}^{p-1} x_{j,j+l} |e_{c(j)} \rangle \langle e_{c(j+l)}|\in B_l$ we have that  
\begin{eqnarray*}
\begin{aligned} 
\Phi_c(x^{(l)})&=\displaystyle \sum_{k=0}^{p-1}\sum_{j=0}^{p-1}\gamma(p-k) x_{j,j+l}|e_{c(j-k)}\rangle \langle e_{c(j+l-k)}| \\ &=\sum_{j=0}^{p-1}\left( \sum_{k=0}^{p-1}\gamma(p-k) x_{j+k,j+k+l}\right)|e_{c(j)}\rangle \langle e_{c(j+l)}|\in B_l.
\end{aligned}
\end{eqnarray*} Moreover, using the isomorphism induced by the cycle $c$ we get 
\begin{eqnarray*}
\begin{aligned}& \Phi_c(|e_{c(j)}\rangle\langle e_{c(j+l)}|)= \sum_{k=0}^{p-1}\gamma(p-k)|e_{c(j-k)}\rangle\langle e_{c(j-k+l)}|  \longmapsto \\ & \sum_{k=0}^{p-1}\gamma(p-k)e_{c(j-k)} = e_{c(j)}\left(\begin{array}{ccccc}  \gamma(0)   &\gamma(1) & \gamma(2) &\cdots &\gamma(p-1)\\ \gamma(p-1) &  \gamma(0)     & \gamma(1) &\cdots & \gamma(p-2) \\ \gamma(p-2) & \gamma(p-1)& \gamma(0) & \cdots &\gamma(p-3)\\ \vdots &   \vdots& \vdots  & \ddots   & \vdots     \\ \gamma(1)    & \gamma(2) & \gamma(3) & \cdots & \gamma(0) \end{array}\right).
\end{aligned}
\end{eqnarray*}

\begin{theorem} For every circulant CP map $\Phi$ on $\mathcal M_p(\mathbb C)$ there exists a cycle $c$ of length $p$  such that $\Phi=\Phi_c$. 
\end{theorem}
\begin{proof} Assume that $\Phi$ is a CP circulant map. By Lemma (\ref{lemma-subspaces}) and condition $(iii)$ we have that $\Phi(|e_{c(j)}\rangle \langle e_{c(j+l)}|)=\sum_{k=0}^{p-1} 	\beta(k) |e_{c(j+k)}\rangle \langle e_{c(j+k+l)}|$ with some $\beta(k)$'s independent of $l$ and positive. On the other side, if $\Phi_{c}(x)=\sum_{k}\beta(p-k)J_{c}^{k}xJ_{c}^{k*}$ we have that $\Phi_c(|e_{c(j)}\rangle\langle e_{c(j+l)}|)= \sum_{k=0}^{p-1}\beta(p-k)|e_{c(j-k)}\rangle\langle e_{c(j-k+l)}|= \sum_{k}\beta(k)|e_{c(j+k)}\rangle\langle e_{c(j+k+l)}|$, sums modulus $p$, therefore  $\Phi=\Phi_c$ on $B_l$ for every $0\leq l\leq p-1$. Hence by condition $(ii)$ we can conclude that $\Phi(x)=\Phi_{c}(x)$ for all  $x\in\mathcal M_p (\mathbb C)$ and this finishes the proof.
\end{proof}

Consider the CP map on  $\mathcal M_p (\mathbb C)\otimes \mathcal M_q (\mathbb C)$ defined by  
\begin{eqnarray}\label{circulant-cp-map} \Phi_*(x) = \sum_{i,j} \alpha(p-i,q-j) (J_{p}^{i}\otimes J_{q}^{j}) x (J_{p}^{i}\otimes J_{q}^{j})^{*},\end{eqnarray} with $\alpha(i,j)\geq 0$ for all $0\leq i\leq p-1$, $0\leq j\leq q-1$ and $J_{s}, \; s=p,q$, the left shift operator. 

Motivated by the above discussion, maps of the class (\ref{circulant-cp-map}) will be called block circulant CP maps. More generally, we call block circulant  CP map to any CP linear combination of tensor products of powers of passage matrices. Restriction of block circulant CP maps to invariant subspaces coincide with block circulant matrices with circulant blocks. 

\begin{theorem}\label{invariant-subspaces}
For every $(k,l)\in{\mathbb Z}_{p}\otimes{\mathbb Z}_{q}$ let $B_{kl}$ be the subspace of $\mathcal M_p (\mathbb C)\otimes \mathcal M_q (\mathbb C)$ defined by  

\[B_{kl}=\textrm{span}\{|e_{i}\rangle\langle e_{i+k}| \otimes |e_{j}\rangle\langle e_{j+l}| : 0\leq i\leq p-1, \, 0\leq j\leq q-1\}.\] Then, 
\begin{itemize}
\item[(i)] the $pq$-dimensional subspaces $B_{kl}$ are mutually orthogonal with the Hilbert-Schmidt product, invariant under the action of $\Phi_{*}$ given by (\ref{circulant-cp-map}), $\bigoplus_{kl}{B_{kl}}={\mathcal M}_{p}(\mathbb C)\otimes{\mathcal M}_{q}(\mathbb C)$ and moreover, 

\item[(ii)] the restriction of $\Phi_{*}$ to any subspace $B_{kl}$ reduces to the action the block circulant matrix $Q=\sum_{i,j}\alpha(i,j)(J_{p}^{i}\otimes J_{q}^{j})$, through the isomorphism from $B_{kl}$ onto ${\mathbb C}^{p}\otimes{\mathbb C}^{q}$ defined by $|e_{i}\otimes e_{j}\rangle\langle e_{i+k}\otimes e_{j+l}|\mapsto e_{i}\otimes e_{j}.$ More precisely, 
\begin{eqnarray*}
\begin{aligned}& \Phi_{*}(|e_{i_0}\rangle\langle e_{i_0 +k}|\otimes |e_{j_0}\rangle \langle e_{j_0 +l}|) \\ &= \sum_{i,j}\alpha\big(p-(i_0-i),q-(j_0-j)\big) |e_{i}\otimes e_{j}\rangle \langle e_{i+k}\otimes e_{j+l}|  \mapsto \\ & \sum_{i,j}\alpha\big(p-(i_0-i),q-(j_0-j)\big)e_{i}\otimes e_{j} = (e_{i_0}\otimes e_{j_0}) Q.
\end{aligned}
\end{eqnarray*} Where $Q$ is the block circulant matrix $Q=\textrm{circ}(Q_0 , Q_1 , \cdots, Q_ {p-1})$ with circulant blocks 
\begin{eqnarray}\label{blocks}
\begin{aligned}& Q_{i}=\left(\begin{array}{ccccc}  \alpha(i,0)   &\alpha(i,1) & \alpha(i,2) &\cdots &\alpha(i,q-1) \\ 
\alpha(i,q-1) &  \alpha(i,0)     & \alpha(i,1) &\cdots & \alpha(i,q-2) \\ 
\alpha(i,q-2) & \alpha(i,q-1)& \alpha(i,0) & \cdots &\alpha(i,q-3)\\ \vdots &   \vdots& \vdots  & \ddots   & \vdots  \\ 
\alpha(i,1)    & \alpha(i,2) & \alpha(i,3) & \cdots & \alpha(i,0) \end{array}\right), \; 0\leq i\leq p-1.
\end{aligned}
\end{eqnarray}
\end{itemize}
\end{theorem}

\begin{proof} 
For every fixed $(i_0, j_0)\in{\mathbb Z}_{p}\times{\mathbb Z}_{q}$ we have that 
\begin{eqnarray*}
\begin{aligned} & \Phi_{*}(|e_{i_0}\rangle\langle e_{i_0 +k}|\otimes |e_{j_0}\rangle \langle e_{j_0 +l}|) \\ &= \sum_{i,j}\alpha(p-i,q-j)\Big(J_{p}^{i}\otimes J_{q}^{j}\Big)|e_{i_0}\otimes e_{j_0}\rangle \langle e_{i_0 +k}\otimes e_{j_0 + l}|\Big(J_{p}^{i}\otimes J_{q}^{j}\Big)^{*} \\ &= \sum_{i,j} \alpha(p-i,q-j) |e_{i_0 -i}\otimes e_{j_0 -j}\rangle \langle e_{i_0 - i+ k}\otimes e_{j_0 - j+l}|  \mapsto \\ & \sum_{i,j} \alpha(p-i,q-j) (e_{i_0 -i}\otimes e_{j_0 -j}) = (e_{i_0}\otimes e_{j_0}) Q.
\end{aligned}
\end{eqnarray*} This proves that every subspace $B_{kl}$ is invariant. They are mutually orthogonal, since 
\begin{eqnarray*}
\begin{aligned} 
tr \Big(|e_{i+k} \otimes e_{j+l}\rangle\langle e_{i}\otimes e_{j}| |e_{i}\otimes e_{j}\rangle\langle e_{i+k'}\otimes e_{j+l'}|\Big) = \delta_{kl, k'l'} .
\end{aligned}
\end{eqnarray*} This proves the Theorem.
\end{proof}

\subsection{Circulant quantum Markov semigroups} 

Consider the discrete time Markov chain on the abelian group ${\mathbb Z}_{p}\times {\mathbb Z}_{q}$ associated with a given probability distribution $\alpha:{\mathbb Z}_{p}\times {\mathbb Z}_{q}\mapsto [0,1]$ with $\sum_{i,j}\alpha(i,j)=1$. If we set $\alpha(0,0)=0$, then the corresponding \textit{bi-stochastic circulant} transition probabilities matrix \[\Pi =\sum_{i,j}\alpha(i,j) (J_{p}^{i}\otimes J_{q}^{j}),\] can be considered as the transition probability  matrix of the embedded Markov chain of the continuous time Markov chain with infinitesimal generator (or Q-matrix) $Q=\Pi-\unit$, where $\unit$ denotes the identity matrix in ${\mathcal M}_{p}(\mathbb C)\otimes {\mathcal M}_{q}(\mathbb C)\cong{\mathcal M}_{n}(\mathbb C)$. Clearly $Q$ is a block circulant matrix with circulant blocks, we shall consider the quantum extensions, in pre-dual representation, 
\begin{eqnarray}\label{cycle-rep-phi} 
\Phi_{*}(x)=\sum_{(i,j)\in{\mathbb Z}_{p}\times{\mathbb Z}_{q}}\alpha(p-i,q-j) (J_{p}^{i}\otimes J_{q}^{j})x(J_{p}^{i}\otimes J_{q}^{j})^{*}.
\end{eqnarray} and 

\begin{eqnarray}\label{cycle-rep-gen}
{\mathcal L}_{*}(x)=\Phi_{*}(x)-x.
\end{eqnarray} of $\Pi$ and $Q$, respectively, with $x\in{\mathcal M}_{p}(\mathbb C)\otimes {\mathcal M}_{q}(\mathbb C)$. Clearly $\Phi_{*}$ is a circulant CP map (\textit{embedded quantum Markov chain}) . We call ${\mathcal L}_{*}$ a \textit{circulant} GKSL generator and \textit{circulant} qms the semigroup generated by ${\mathcal L}_{*}$.

Instead of the matrices $J_{p}$ and $J_{q}$, we can choose any pair of passage matrices $J_{c_p},J_{c_q}$ of cycles of maximal length in ${\mathbb Z}_{p}$ and ${\mathbb Z}_{q}$ respectively, having $G=\mathbb Z_p \times \mathbb Z_q$. Even more, any finite number $n$ of maximal length cycles $c_k$ can be chosen in $\mathbb Z_{p_k}$ respectively, and follow the computation along the same lines with $G=\times_{k=0}^{n-1} \mathbb Z_{p_k}$. 
Moreover, if ${\mathbb Z}_{p}$ has a prime order, then every power $J_{c}^k$ of a passage matrix $J_c$ is the passage matrix of some cycle $c_{k}\neq c$ if $0\neq k\neq 1$ (mod p). 
   
Having this in mind and Kalpaziduo's cycle representation\cite{kalpaz} of an irreducible Markov chain with uniform stationary measure $\pi=\{\frac{1}{p}\}$ and circulant generator Q,  \begin{align*}\frac{1}{p}Q= \sum_{c\in\mathcal C_\infty}  w_{c} J_c - \frac{1}{p}I, \end{align*}  we can regard equations (\ref{cycle-rep-phi}) and (\ref{cycle-rep-gen}) as  a quantum cycle representation of the  \textit{circulant} GKSL generator ${\mathcal L}_{*}$ with cycle weights $(\alpha(i,j))_{(i,j)\in{\mathbb Z}_{p}\times{\mathbb Z}_{q}}$. This motivates the following.

\begin{definition}\label{def-cycle-rep} 
Given a bounded GKSL generator of the form (\ref{gksl-rep}) with a dis\-cre\-te spectrum Hamiltonian, we call \textit{cycle representation} of its embedded quantum Markov chain \[\Phi(x)=\sum_{k} L_{k}^{*} x L_{k},\] to a GKSL representation of $\Phi$ of the form \[\Phi(x)=\sum_{l}\alpha_{l} U_{l}^{*} xU_{l},\]  where for each $l$, $\alpha_{l} >0$ and $U_{l}$ is a  passage matrix.
\end{definition} Clearly, any cycle decomposition of the embedded chain $\Phi$ induces a cycle representation of ${\mathcal L}$. 

\begin{remark}
Tensor product like $J_{p}^{i}\otimes J_{q}^{j}$ are irreducible matrices. Hence by Lemma \ref{lema1}, they are passage matrices of a cycle. This shows that our definition includes representations of the form (\ref{cycle-rep-phi}) and its extensions involving any finite number of cycles or higher order tensor products.
\end{remark}

By Theorem \ref{invariant-subspaces} each subspace $B_{kl}$ is invariant under $\Phi_*$, ${\mathcal L}_{*}$ and, consequently, also under the action of the semigroups ${\mathcal T}_{*}=\big({\mathcal T}_{*t}\big)_{t\geq 0}$ generated by ${\mathcal L}_{*}$. The state $\rho=\frac{1}{pq}(\unit\otimes\unit)$ is clearly invariant for  ${\mathcal T}_{*}$ since \[{\mathcal L}_{*}(\rho)=\sum_{(i,j)\neq(0,0)}\alpha(p-i,q-j)\unit\otimes\unit - \unit\otimes\unit=0,\] because $\sum_{(i,j)\neq(0,0)}\alpha(p-i,q-j)=1$.The $\rho$-adjoint $\tilde{\mathcal L}$ has the GKSL representation \[\tilde{\mathcal L}(x)=\sum_{i,j}\tilde{L}_{ij}^{*}x\tilde{L}_{ij},\] with $\tilde{L}_{ij}=L_{ij}^{*}$ and $L_{ij}=\alpha(p-i,q-j)^{\frac{1}{2}}(J_{p}^{i}\otimes J_{q}^{j})$. Hence $\tilde{L}_{ij}=\alpha(i,j)^{\frac{1}{2}}(J^i_p\otimes J^j_q)$.

One can write in direct representation \[\tilde{\mathcal L}(x)=\sum_{(i,j)\neq(0,0)}\alpha(i,j)(J_{p}^{i}\otimes J_{q}^{j})^{*}x(J_{p}^{i}\otimes J_{q}^{j}) - x.\] Hence the difference between ${\mathcal L}$ and its $\rho$-adjoint (reverse) operator looks like 
\begin{eqnarray*}
\begin{aligned}
\tilde{\mathcal L}(x)-{\mathcal L}(x)&=\sum_{(i,j)\neq(0,0)}\big(\alpha(i, j)-\alpha(p-i,q-j)\big)(J_{p}^{i}\otimes J_{q}^{j})^{*}x(J_{p}^{i}\otimes J_{q}^{j}) \\ & = \sum_{(i,j)\neq(0,0)}\big(q(i,j)-1\big) {L}^{*}_{ij}x{L}_{ij},
\end{aligned}
\end{eqnarray*} with $q(i,j)=\alpha(i,j){\alpha(p-i,q-j)}^{-1}$ and the ${L}_{ij}$'s as above. Therefore, the semigroup ${\mathcal T}_{*}$ satisfies a weighted detailed balance condition in the sense of Accardi-Fagnola-Quezada\cite{a-f-q} with weights $q=\big(q_{ij}=\alpha(i,j){\alpha(p-i,q-j)}^{-1}\big)$, see equation (\ref{def-weighted-db}) above. Consequently, by Corollary 2 in Ref.\cite{a-f-q}, detailed balance holds if and only if 
\begin{eqnarray}\label{qdb-alphas}
\alpha(p-i,q-j)=\alpha(i,j), \; \forall \; (0,0)\neq(i,j)\in{\mathbb Z}_{p}\times{\mathbb Z}_{q}.
\end{eqnarray}

\section{Quantum Entropy Production Rate for circulant qms}\label{QEPR}
Let us compute the Quantum Entropy Production Rate (QEPR) for the circulant semigroup ${\mathcal T}_{*}$ in the previous section. For simplicity we consider first the invariant state $\rho=\frac{1}{pq} \unit$, other invariant states are studied in Section \ref{other-invariant-states}. We know that every subspace $B_{kl}$ of ${\mathcal M}_{p}\otimes{\mathcal M}_{q}$ is invariant under the action of the elements of ${\mathcal T}_{*}$. This implies that the states $\Omega_{t}$ and $\tilde{\Omega}_{t}$ are diagonal with respect to the canonical basis.  

\begin{lemma}\label{diag-forward-qms}
With the notations in Section \ref{circ-qms} and Subsection \ref{entropy-pro-qms} the following hold:
\begin{itemize}
\item[(i)] for every $(i,j), (i', j') \in{\mathbb Z}_{p}\times{\mathbb Z}_{q}$, using the isomorphism induced by lemma \ref{lema1}, we have
\begin{eqnarray}
\begin{aligned}
 &{\mathcal T}_{*t}(|e_{i}\otimes e_{j}\rangle \langle e_{i'}\otimes e_{j'}|)\mapsto (e_{i}\otimes e_{j})e^{tQ} = \frac{1}{pq} \sum_{m,n} \Phi_{m-i,n-j}(t)(e_{m}\otimes e_{n}) \\ & \mapsto \frac{1}{pq}\sum_{m,n}\Phi_{m,n}(t) |e_{m+i}\otimes e_{n+j}\rangle\langle e_{m+i'}\otimes e_{n+j'}|. 
 \end{aligned}
\end{eqnarray} where \[\Phi_{m,n}(t)= \sum_{k,l}\omega_{p}^{mk}\omega_{q}^{nl}e^{t\lambda_{kl}}, \; \; \lambda_{kl}=\sum_{i,j}\alpha(i,j)\overline{\omega}_{p}^{ik}\overline{\omega}_{q}^{jl}.\] We recall that in this case $\alpha(0,0)=-1$ and $\sum_{(i,j)\neq (0,0)}\alpha(i,j)=1$. Moreover, the functions $\Phi_{m,n}(t)$ are real-valued, since $Q$ and hence $e^{tQ}$ are real matrices. 
\item[(ii)] 
\begin{eqnarray}
\begin{aligned}
\Omega_{t}= \frac{1}{pq} \sum_{m,n}\Phi_{m,n}(t) |u_{mn}\rangle\langle u_{mn}|,
\end{aligned}
\end{eqnarray} where $u_{mn}= \displaystyle\sum_{ij} \rho_{ij}^{\frac{1}{2}}(e_{i}\otimes e_{j})\otimes(e_{m+i}\otimes e_{n+j})$.
\end{itemize}
\end{lemma}  

\begin{proof}
Item $(i)$ is an immediate consequence of Theorem \ref{diag-block-circ} and Corollary \ref{exp-block-circ}.  Now a direct computation using $(i)$ shows that 
\begin{eqnarray}
\begin{aligned}
&\Omega_{t} = \\ &\frac{1}{pq} \sum_{m,n}\Phi_{m,n}(t)\Big| \sum_{i,j}\rho^\frac{1}{2}_{ij} e_{i}\otimes e_{j}\otimes e_{m+i}\otimes e_{n+j}\Big\rangle\Big\langle \sum_{r,s}\rho_{rs}^\frac{1}{2} e_{r}\otimes e_{s}\otimes e_{m+r}\otimes e_{n+s} \Big| \\ & 
=\frac{1}{2} \frac{1}{pq} \sum_{m,n}\Phi_{m,n}(t)|u_{mn}\rangle \langle u_{mn}|.
\end{aligned}
\end{eqnarray} This finishes the proof.
\end{proof}

The subspaces $B_{kl}$ are invariant also for the reverse semigroup $\tilde{\mathcal T}_{*t}$. Moreover, similar computations yield the following.

\begin{lemma}\label{diag-backward-qms} For the $\rho$-adjoint (reverse) semigroup we have: 
\begin{itemize}
\item[(i)]
\begin{eqnarray}
\begin{aligned}
& \tilde{\mathcal T}_{*t}(|e_{i}\otimes e_{j}\rangle \langle e_{i'}\otimes e_{j'}|)\mapsto (e_{i}\otimes e_{j})e^{tQ^{*}} = \frac{1}{pq} \sum_{m,n} \tilde{\Phi}_{m-i,n-j}(t)(e_{m}\otimes e_{n}) \\ & \mapsto \frac{1}{pq}\sum_{m,n}\tilde{\Phi}_{m,n}(t) |e_{m+i}\otimes e_{n+j}\rangle\langle e_{m+i'}\otimes e_{n+j'}|, 
\end{aligned}
\end{eqnarray} where $Q^{*}$ is the transpose of  $Q$ and $\tilde{\Phi}_{m,n}=\Phi_{p-m,q-n}$.
\item[(ii)] 
\begin{eqnarray}
\begin{aligned}
\tilde{\Omega}_{t}= \frac{1}{pq}\sum_{m,n}{\Phi}_{p-m,q-n}(t) |u_{mn}\rangle\langle u_{mn}|,
\end{aligned}
\end{eqnarray} 
\end{itemize}
\end{lemma}

\begin{theorem}\label{entropy-pro-expression}
Let ${\mathcal L}_{*}$ be a circulant GKSL generator of the form (\ref{cycle-rep-gen}), then the Quantum Entropy Production Rate of the corresponding qms is given by 
\[e_p({\mathcal T_*}, \rho)=\frac{1}{2} \frac{1}{pq} \sum_{(m,n)\in{\mathbb Z}_{p}\times{\mathbb Z_{q}}} \big(\alpha(m,n)- \alpha(p-m,q-n)\big) \log\frac{ \alpha(m,n)}{\alpha(p-m,q-n)}.\]
\end{theorem}

\begin{proof}
From the above lemmata it follows that the  relative entropy has the explicit expression,   
\begin{eqnarray}
\begin{aligned}S(\Omega_t,\tilde{\Omega}_t)&= \frac{1}{pq} \sum_{m,n}\Phi_{m,n}(t) log\frac{\Phi_{m,n}(t) }{{\Phi}_{p-m,q-n}(t) }\\ &=\frac{1}{2} \frac{1}{pq} \sum_{m,n}\Big( \Phi_{m,n}(t)-{\Phi}_{p-m,q-n}(t)\Big) log\frac{\Phi_{m,n}(t)}{{\Phi}_{p-m,q-n}(t)}. 
\end{aligned}
\end{eqnarray}

For the Quantum Entropy Production Rate we have, 
\begin{eqnarray*}
\begin{aligned} 
e_p({\mathcal T_*}, \rho) &=\lim_{t\to 0^+} \frac{S(\Omega_t,\tilde{\Omega}_t)}{t} \\ &=\frac{1}{2}  \frac{1}{pq} \sum_{m,n}\left(\lim_{t\to 0^+}\frac{ \Phi_{m,n}(t)-{\Phi}_{p-m,q-n} (t)}{t}\right) \lim_{t\to 0^+}log\frac{\frac{\Phi_{m,n}(t)}{t}}{ \frac{{\Phi}_{p-m,q-n}(t)}{t}}.  \end{aligned}
\end{eqnarray*}

But a simple computation shows that for every $m,n$, 
\begin{eqnarray*}
\begin{aligned}
\lim_{t\to 0^+}\frac{\Phi_{m,n}(t)}{t}&=\lim_{t\to 0^+}\Big( \left\langle e_{0}\otimes e_{0},\frac{e^{tQ}-I}{t} e_{m}\otimes e_{n} \right\rangle + \frac{\langle e_{0}\otimes e_{0}, e_{m}\otimes e_{n} \rangle}{t}\Big)\\ &=\langle e_{0}\otimes e_{0},Q (e_{m}\otimes e_{n}) \rangle=\alpha(m,n). 	
\end{aligned}
\end{eqnarray*} Therefore 

\begin{align*}e_p({\mathcal T_*}, \rho)=\frac{1}{2} \frac{1}{pq}  \sum_{m,n}\Big( \alpha(m,n)-\alpha(p-m,q-n)\Big) log\frac{\alpha(m,n)}{\alpha(p-m,q-n)}.
\end{align*} This finishes the proof. 
\end{proof}

\section{Comparison to Classical Entropy Production Rate}\label{QEPR-EPR-compare}

The Quantum Entropy Production Rate (\ref{quan-ent}) aims at generalizing the classical one, hence it is natural to expect that some relation can be found between them. In this section we compute explicitly the (classical) Entropy Production Rate for the restriction of Circulant Quantum Markov Semigroups to the diagonal commutative sub-algebra, namely $B_{00}$, and show it actually coincides with its quantum counterpart. 

According to Qian et al.\cite{qian}, the Classical Entropy Production Rate of an irreducible Markov chain with intensity matrix $Q=(q_{ij})_{i,j\in S}$ and stationary measure $\pi=(\pi_i)_{i\in S}$, over a finite state space $S$ is given by 
\begin{eqnarray}\label{qians-formula}
\begin{aligned} e_p=\frac{1}{2}\displaystyle \sum_{i,j\in S} (\pi_i q_{ij}-\pi_{j} q_{ji}) \log \dfrac{\pi_i q_{ij}}{\pi_j q_{ji}}.
\end{aligned}
\end{eqnarray} 

By Theorem \ref{invariant-subspaces}, the restriction of $\mathcal L_*$ to $B_{00}$ reduces to the action of the block circulant matrix $Q=circ(Q_0,Q_1,\ldots,Q_{p-1})$, with circulant blocks of the form (\ref{blocks}) and $\alpha(0,0)=-1$. In terms of the distribution $\alpha$, each matrix element of $Q$ is given by $q_{ij}= \alpha\big((l-k)_p , (r_j -r_i)_q\big)$  where for every pair $0\leq i,j \leq pq-1$ we write $i=qk+r_i$, $ j=ql+r_j$, $0\leq k,l\leq p-1$, $0\leq r_i,r_j \leq q-1$, and for every $-(s-1)\leq x\leq s-1$, $s=p,q$, we define 

\begin{align*}\begin{array}{cc} (x)_s =\left\{\begin{array}{cc}x & \mbox{\text{ if $x\geq0$}} \\
                                             s+x &\mbox{\text{ if $x<0$}}.\end{array} \right.
                                             \end{array}  \end{align*}
Clearly the relation $(-x)_s=s-(x)_s$ holds true.  
 \begin{corollary}The Quantum Entropy Production Rate of a Circulant qms equals the Classical Entropy Production Rate of its diagonal-restricted Markov chain, i.e.,  
 \begin{align*}e_p({\mathcal T_*}, \rho)=e_p. \end{align*}
\end{corollary}
 
 \begin{proof} An application of the above formula (\ref{qians-formula}), re-ordering the sum according with the order of  blocks and the change of variables $m=(l-k)_p$, $n=(r_j-r_i)_q$ yields,   
 
 \begin{eqnarray*}
 \begin{aligned}
e_p  = & \frac{1}{2} \frac{1}{pq} \displaystyle \sum_{k,l=0}^{p-1}\sum_{r_i,r_j=0}^{q-1} \Big(\alpha\big((l-k)_p , (r_j -r_i)_q\big)-\alpha\big((k-l)_p , (r_i -r_j)_q\big)\Big) \times \\ & \log \frac{\alpha\big((l-k)_p , (r_j -r_i)_q\big)}{\alpha\big((k-l)_p , (r_i -r_j)_q\big)} \; 
 \\&= \; \frac{1}{2} \frac{1}{pq}\displaystyle \sum_{(m,n)\in{\mathbb Z}_{p}\times{\mathbb Z_{q}}}\Big(\alpha(m,n)-\alpha(p-m,q-n)\Big)  \log\frac{\alpha(m,n)}{\alpha(p-m,q-n)}\\&=e_p({\mathcal T_*}, \rho).
\end{aligned}
\end{eqnarray*} This proves the corollary.
\end{proof}
 
\section{QEPR with respect to other invariant states}\label{other-invariant-states}

To close the paper, in this section we compute the QEPR in any invariant state of the semigroup ${\mathcal T}_{*}$. 

\begin{proposition} Every invariant state of $\mathcal L_*$ has the form \begin{eqnarray}\label{invariant-state}\rho=\frac{1}{pq}\unit_p \otimes \unit_q +\displaystyle \sum_{ij} \rho_{ij} J^i_p \otimes J^j_q , \end{eqnarray} where $\rho_{ij}$ are complex numbers constrained by the positiveness of $\rho$. \end{proposition}

\begin{proof} We decompose $\rho$ into its mutually orthogonal components in the subspaces $B_{kl}$, namely $\rho=\sum_{kl}\hat{\rho}_{kl}$. Clearly $\mathcal L_*(\rho)=0$ if and only if ${\mathcal L}(\hat{\rho}_{kl})=0$ for every $(k, l)\in{\mathbb Z}_{p}\times{\mathbb Z}_{q}$. As a consequence of Theorem \ref{invariant-subspaces}, using the isomorphism defined there, each of the above conditions becomes a linear system of equations of the form $\hat{\rho}_{kl}Q=0$, where $Q$ is the same circulant matrix for all systems. Any solution to these systems is a multiple of  the identity vector, which yields the solution (\ref{invariant-state}). Although every choice of complex constants $\rho_{kl}$ give a solution of $\mathcal L_*(\rho)=0$, not all of them give back a state $\rho$. In fact, $\rho_{00}=\frac{1}{pq}$ so that $tr\rho=1$ while the remaining $\rho_{kl}'s$ are constrained by the positiveness of $\rho$. 
Conversely, if $\rho$ has the form (\ref{invariant-state}) then  $\mathcal L_*(\rho)=\rho\left( \displaystyle \sum_{ij\neq 0} \alpha(p-i,q-j)\unit_p \otimes \unit_q -  \unit_p \otimes \unit_q \right)=0.$
\end{proof}

By Lemma \ref{diagonal-circ}, any invariant state $\rho$ can be diagonalized by the discrete Fourier Transform, indeed, 
\[\rho=\sum_{lk} \tilde{\rho}_{kl} |\tilde{e_l} \otimes \tilde{e_k} \rangle \langle \tilde{e_l} \otimes \tilde{e_k} |,\] where $\tilde{\rho}_{lk}=\dfrac{1}{pq}+\displaystyle \sum_{ij} \rho_{ij} \overline{\omega}_{p}^{ik}\overline{\omega}_{q}^{jl}$ and $\tilde{e}_l=F_{p}^{*}e_{l}$, $\tilde{e}_k=F_{q}^{*}e_{k}.$

In the next computations it is understood that sums over the first coordinate of the tensor product go from 0 to $p-1$ and sums over the second coordinate go from $0$ to $q-1$. We use the results and notations in Lemma \ref{diag-forward-qms}.

Let us compute the state associated with ${\mathcal T}_{*}$ using the basis of $\rho$, $\{\tilde{e}_{i}\otimes \tilde{e}_{j}\}_{(i,j)\in{\mathbb Z}_{p}\times{\mathbb Z}_{q}}$,

\begin{eqnarray*}
\begin{aligned}  \Omega_t= & \displaystyle \sum_{ii^\prime j j^\prime} |\tilde{e}_i \otimes \tilde{e}_{i^{\prime}} \rangle \langle \tilde{e}_j \otimes \tilde{e}_{j^{\prime}}| \otimes \mathcal T_{*t} (\rho^\frac{1}{2} |\tilde{e}_i \otimes \tilde{e}_{i^{\prime}} \rangle \langle \tilde{e}_j \otimes \tilde{e}_{j^{\prime}}|\rho^\frac{1}{2} ) \\
= &\frac{1}{pq}\displaystyle \sum_{\substack{ ii^\prime jj^\prime \\n n^\prime  r r^\prime} } \tilde{\rho}^{\frac{1}{2}}_{ii^\prime}\tilde{\rho}^{\frac{1}{2}}_{jj^\prime} \overline{\omega}_{p}^{ni} \overline{\omega}_{q}^{n^\prime i^\prime} \omega_{p}^{rj} \omega_{q}^{r^\prime j^\prime} |e_n \otimes e_{n^\prime} \rangle \langle e_r \otimes e_{r^\prime} | \otimes \mathcal T_{*t} ( |\tilde{e}_i \otimes \tilde{e}_{i^{\prime}} \rangle \langle \tilde{e}_j \otimes \tilde{e}_{j^{\prime}}| ) \\
=&\frac{1}{(pq)^2}\displaystyle \sum_{\substack{ ii^\prime jj^\prime \\n n^\prime  r r^\prime \\ NN^\prime RR^\prime}}\tilde{\rho}^{\frac{1}{2}}_{ii^\prime}\tilde{\rho}^{\frac{1}{2}}_{jj^\prime} \overline{\omega}_{p}^{(n+N)i}\overline{\omega}_{q}^{(n^\prime+N^\prime)i^\prime}  \omega_{p}^{(r+R)j} \omega_{q}^{(r^\prime+R^\prime)j^\prime}\times \\ & |e_n \otimes e_{n^\prime} \rangle \langle e_r \otimes e_{r^\prime} | \otimes \mathcal T_{*t}\left( |e_N \otimes e_{N^\prime} \rangle \langle e_R \otimes e_{R^\prime} | \right ) \\
=&\frac{1}{pq}\sum_{\substack{n n^\prime  r r^\prime \\ NN^\prime RR^\prime\\M M^\prime}} \Phi_{M,M^\prime}(t) \beta(n,N,n^\prime,N^\prime) \times \\ &\beta(r,R,r^\prime,R^\prime)|e_n \otimes e_{n^\prime} \rangle \langle e_r \otimes e_{r^\prime} |  \otimes |e_{N+M} \otimes e_{N^\prime+M^\prime} \rangle \langle e_{R+M} \otimes e_{R^\prime+M^\prime} \\
=&\frac{1}{pq}\sum_{mm^\prime}\Phi_{m,m^\prime}(t)|u_{mm^\prime}\rangle \langle u_{mm^\prime}|, 
\end{aligned}
\end{eqnarray*} where  $u_{mm^\prime}=\frac{1}{\sqrt{pq}} \sum_{ll^\prime LL^\prime} \beta(l,L,l^\prime,L^\prime) ( e_l \otimes e_{l^\prime}) \otimes	(e_{L+m}\otimes e_{L^\prime +m^\prime}),$ and $\beta(l,L,l^\prime,L^\prime)=  \frac{1}{\sqrt{pq}}\sum_{ii^\prime}\tilde{\rho}^{\frac{1}{2}}_{ii^\prime} \overline{\omega}_{p}^{(l+L)i}\overline{\omega}_{q}^{(l^\prime+L^\prime)i^\prime}.$

Direct computations show that the $\rho$-adjoint semigroup, with respect to any $\rho$ of the form (\ref{invariant-state}),  coincide with $\tilde{\mathcal T}_{*}$ given by Lemma \ref{diag-backward-qms}. In a similar way we get 
\begin{eqnarray}
\begin{aligned}
\tilde{\Omega}_{t}= \frac{1}{pq}\sum_{m,m'}{\Phi}_{p-m,q-m'}(t) |u_{mm'}\rangle\langle u_{mm'}|.
\end{aligned}
\end{eqnarray}  It follows that the Quantum Entropy Production Rate in any invariant state $\rho$ of the form (\ref{invariant-state}) coincides with the one given by Theorem \ref{entropy-pro-expression}.

\begin{theorem}\label{equiv-db-ep=0}
Let ${\mathcal T}_{*}$ a circulant qms with GKSL generator $\mathcal L_*$ of the form (\ref{cycle-rep-gen}), then the following are equivalent: 
\begin{itemize}

\item[(i)] ${\mathcal T}_*$ satisfies a quantum detailed balance condition with respect to any invariant state $\rho$ of the form (\ref{invariant-state}), 

\item[(ii)] $\alpha(m,n)=\alpha(p-m, q-n)$ for all $(m,n)\in{\mathbb Z}_{p}\times {\mathbb Z}_{q}$, 

\item[(iii)] the Quantum Entropy Production Rate of ${\mathcal T}_{*}$ with respect to any stationary state  $\rho$ of the form (\ref{invariant-state}) equals zero, i.e., $e_{p}({\mathcal T}_{*}, \rho) = 0.$
\end{itemize}
\end{theorem}

\begin{proof}
The equivalence of $(i)$ and $(ii)$ follows from Corollary 2 in Ref.\cite{a-f-q}, see (\ref{qdb-alphas}). And the equivalence of $(ii)$ with $(iii)$ follows from Theorem \ref{entropy-pro-expression}. 
\end{proof}

\begin{remark}
\begin{itemize}
\item[(i)] Theorems \ref{entropy-pro-expression} and \ref{equiv-db-ep=0} have a direct generalization to the case of any finite number of cycles (or cyclic factors in the abelian group $G$). 
\item[(ii)] We remark that in the case of a separable probability distribution $\alpha(i,j)=\alpha_{p}(i)\alpha_{q}(j)$, a direct computation using Lemmata \ref{diag-forward-qms}, \ref{diag-backward-qms} shows that the states $\Omega_{t}, \; \tilde{\Omega}_{t}$ are separable. Indeed, $\Omega_{t}=\Omega_{p}(t)\otimes \Omega_{q}(t)$, with 
\[\Omega_{s}(t)=\frac{1}{s}\sum_{i}\Phi_{s}(i,t)|u_{s}(i)\rangle\langle u_{s}(i)|,\] where $\Phi_{s}(i,t)=\displaystyle\sum_{j}\omega^{ij} e^{t\lambda_{j}(\alpha_{s}})$, \; $\displaystyle\lambda_{j}(\alpha_{s})=\sum_{l}\alpha_{s}(i)\overline{\omega}^{lj}$, \; \; \textrm{and} \; \; $u_{s}(k)=\displaystyle\frac{1}{\sqrt{s}}\sum_{n}|e_{n}\rangle\langle e_{n+k}|$, $s=p,q$. 
\end{itemize}
\end{remark}

\section*{Acknowledgement} The financial support from CONACYT-Mexico and Ministero degli Affari Esteri-Italy, through the joint research project ``Din\'amica Estoc\'astica con Aplicaciones en F\'isica y Finanzas", is gratefully acknowledged.

\end{document}